\documentclass[12pt,reqno]{amsart}
\usepackage{amsmath,amsfonts,amsthm,amssymb,mathabx}
\newcommand\version{Feb 11, 2019}
\usepackage{amsxtra}

\newtheorem{theorem}{Theorem}

\newtheorem{lemma}{Lemma}

\theoremstyle{definition}

\theoremstyle{remark}

\usepackage{tikz}

\numberwithin{equation}{section}

\renewcommand{\epsilon}{\varepsilon}

\newcommand{\F}{\mathcal{F}}

\newcommand{\N}{\mathbb{N}}

\newcommand{\R}{\mathbb{R}}

\DeclareMathOperator{\infspec}{inf spec}

\begin{document}

\title[Polaron mass at strong coupling --- \version]{Divergence of the effective mass of a polaron  in the strong coupling limit}

\author{Elliott H. Lieb}
\address{(E.H. Lieb) Departments of Mathematics and Physics, Princeton University, Princeton, NJ 08544, USA}
\email{lieb@princeton.edu}

\author{Robert Seiringer}
\address{(R. Seiringer) IST Austria, Am Campus 1, 3400 Klosterneuburg, Austria}
\email{rseiring@ist.ac.at}

\thanks{\copyright\, 2019 by the authors. This paper may be  
reproduced, in
its entirety, for non-commercial purposes.}

\begin{abstract} We consider the Fr\"ohlich model of a polaron, and show that its effective mass diverges in the strong coupling limit.
\end{abstract}

\date{\version}

\maketitle

\section{Introduction and main result}
The polaron model introduced by Fr\"ohlich  \cite{Fr} represents a simple and well-studied model of an electron interacting with the quantized optical modes of a polar crystal. We refer to    
 \cite{devreese,FLST,GL,Mo,spohn} for properties, results and further references. To this date, the asymptotic behavior of its effective mass for strong coupling represents an outstanding open problem. According to Landau and Pekar \cite{pekar}, it is expected to diverge as $\alpha^4$ for large coupling constant $\alpha$, with a prefactor determined by the minimizer of the Pekar functional, see Eqs.~\eqref{def:EP} and~\eqref{m:exp} below. While we are not able to verify this conjecture, we shall prove in this paper that the effective mass indeed diverges to infinity as $\alpha\to \infty$. 

For fixed total momentum $P\in \R^3$, the Hamiltonian of the Fr\"ohlich model is given by \cite{LLP,Mo} 
\begin{equation}
H_P = (P- P_f)^2 + V + \N
\end{equation}
where
\begin{equation}\label{def:V}
V = - \frac{\sqrt{\alpha}}{\sqrt{2}\pi}  \int_{\R^3} dk\,  \frac 1{|k|} \left( a_k + a^\dagger_k \right) \,,
\end{equation}
$\N = \int_{\R^3} dk\, a^\dagger_k a_k$ denotes the number operator,  $P_f= \int_{\R^3} dk\, k \, a^\dagger_k a_k$ the field momentum, and $\alpha>0$ is a coupling constant. 
The Hamiltonians $H_P$ act on the Hilbert space $\mathcal{F}$, the bosonic Fock space over $L^2(\R^3)$. The creation and annihilation operators satisfy the usual canonical commutation relations $[a_k,a^\dagger_{l}] = \delta(k-l)$. 

We denote $E_P = \infspec H_P$. It is well-known that $\min_P E_P = E_0$ \cite{GL}, and that
\begin{equation}\label{plim}
\lim_{\alpha\to \infty} \alpha^{-2} E_0 = e^{\rm P}\,,
\end{equation}
with $e^{\rm P}$ the Pekar energy
\begin{equation}\label{def:EP}
e^{\rm P} = \min_\psi \left\{ \int_{\R^3} dx\, |\nabla\psi(x)|^2 - \iint\nolimits_{\R^3\times\R^3} dx\, dy\, \frac{ |\psi(x)|^2 |\psi(y)|^2 }{|x-y|} \, : \, \int_{\R^3}  |\psi|^2 =1 \right\}\,.
\end{equation}
This was proved in \cite{DoVa} using the path-integral formulation of the problem (see also  \cite{MV1,MV2} for recent work on the construction of the Pekar process \cite{spohn}), and quantitative bounds were given later in \cite{LT} using operator methods, which will play an important role also in this work. 

The effective mass $m$ of the polaron is defined via
\begin{equation}
E_P = E_0 + \frac {P^2}{ 2m }  +  o(P^2)
\end{equation}
as $P\to 0$. It satisfies $m\geq 1/2$, which is the bare mass of the electron in our units. In fact, $m>1/2$ for $\alpha>0$. Our goal is to prove

\begin{theorem}\label{main:thm}
The effective mass of the polaron satisfies
\begin{equation}
\lim_{\alpha\to \infty} m =  \infty\,.
\end{equation}
\end{theorem}

According to  \cite{pekar} (see also   \cite{devreese} and \cite{spohn}) the polaron mass is expected to satisfy
\begin{equation}\label{m:exp}
\lim_{\alpha\to\infty} \alpha^{-4} m 
= \frac {8\pi }3 \int_{\R^3} dx\, |\psi^{\rm P}(x)|^4 
\end{equation}
where $\psi^{\rm P}$ denotes the minimizer of the Pekar functional in \eqref{def:EP}. The latter is unique up to translations and multiplication by a complex phase \cite{Li}. While our result is far from showing \eqref{m:exp}, it gives for the first time a lower bound on $m$ that diverges as $\alpha\to \infty$.


To prove Theorem~\ref{main:thm}, we shall compute an upper bound on $E_P - E_0$. The choice of trial state is motivated by the following observation. In the strong coupling limit, we expect \cite{LS14,nagy} the  ground states $\phi_P$ of $H_P$ to be approximately of the form
\begin{equation}\label{ansatz}
\phi_P \, \propto \,  \widehat\psi_\alpha^{\rm P} ( P-P_f  ) e^{a^\dagger(\varphi^{\rm P}_\alpha)} \Omega
\end{equation}
where $\Omega\in \mathcal{F}$ denotes the Fock space vacuum,  $\widehat\psi^{\rm P}_\alpha(p) = \alpha^{-3/2} \widehat\psi^{\rm P}(\alpha^{-1} p) $  is the Fourier transform of a minimizer of the Pekar functional in \eqref{def:EP} with coupling constant $\alpha$ inserted in front of the second term, and $\varphi^{\rm P}_\alpha $ is the corresponding Pekar field function in momentum space, given by 
\begin{equation}\label{def:varphi}
\varphi^{\rm P}_\alpha (p) = \frac{\sqrt{\alpha}}{\sqrt{2}\pi |p|}  \int_{\R^3} dx \, |\psi_\alpha^P(x)|^2 e^{-i p \cdot x}  \,.
\end{equation}
Moreover, $a^\dagger(\varphi^{\rm P}_\alpha)$ is short for
\begin{equation}\label{adnot}
a^\dagger(\varphi^{\rm P}_\alpha) = \int_{\R^3} dk\, \varphi^{\rm P}_\alpha(k) a^\dagger_k \,,
\end{equation}
hence $e^{a^\dagger(\varphi^{\rm P}_\alpha)} \Omega$ is proportional to the coherent state whose expectation of $a_k$ gives $\varphi^{\rm P}_\alpha(k)$. 

In particular, 
we expect that $\phi_P \approx \widehat\psi^{\rm P}_\alpha(P-P_f)/ \widehat\psi_\alpha^{\rm P}(-P_f) \phi_0$, which to leading order in $P$ reads 
\begin{equation}\label{compe}
\phi_P \approx \phi_0 + P\cdot \frac{\nabla \widehat\psi_\alpha^{\rm P}(-P_f) }{\widehat\psi_\alpha^{\rm P}(-P_f)} \phi_0\,.
\end{equation}  
Our actual choice of trial state will be slightly modified, since we do not know whether the function  $p\mapsto \nabla \widehat\psi_\alpha^{\rm P}(-p) /\widehat\psi_\alpha^{\rm P}(-p)$ is bounded, and hence we will use a regularized version of it.

Our method of proof is in principle quantitative, i.e., gives a  lower bound on the effective mass $m$, except for the regularization just mentioned. If one can show that $p\mapsto \nabla \widehat\psi_\alpha^{\rm P}(-p) /\widehat\psi_\alpha^{\rm P}(-p)$  is a bounded function (or get a control on its possible divergence at infinity), one obtains an explicit lower bound on the rate of divergence of $m$ as $\alpha\to \infty$.  Due to the rather crude energy estimates involved, the lower bound is at best of order $\alpha^{1/10}$, however.  This is far from the expected $\alpha^4$ in \eqref{m:exp}. 

In the remainder of this paper we shall give the proof of Theorem~\ref{main:thm}.

\section{Proof of Theorem~\ref{main:thm}}

Let  $\phi_0\in \mathcal{F}$ denote the normalized ground state of $H_0$. Existence and uniqueness of $\phi_0$ are shown in 
\cite{Mo}.\footnote{Strictly speaking, the results in \cite{Mo} apply only to the model with an ultraviolet cutoff. The latter can be removed by a suitable limit, as explained in detail in  \cite{GW}.} Let $t:\R^3 \to \R^3$ be smooth and compactly supported,  and of the form $t(p) = p h(p)$ with $h$ a radial function. 
We take as  trial function for $E_P = \infspec H_P$ a function of the form
\begin{equation}
\phi_P = \phi_0 - \alpha^{-1}P \cdot t(P_f/\alpha) \phi_0 \,.
\end{equation}
Using rotation invariance of $\phi_0$, we see that the norm of $\phi_P$ equals
\begin{equation}\label{norm}
\| \phi_P\|^2 = 1 +  \frac {P^2} {3\alpha^2} \left\langle\phi_0 \left|  |t(P_f/\alpha)|^2 \right| \phi_0\right\rangle\,.
\end{equation}
Moreover,  since $H_0 \phi_0 = E_0 \phi_0$, we also have
\begin{align}\nonumber
\left\langle \phi_P \left| H_P \right| \phi_P\right\rangle & = E_0 + P^2 + \alpha^{-2} \left\langle P\cdot  t(P_f/\alpha)  \phi_0 \left| H_0 \right| P\cdot t(P_f/\alpha) \phi_0 \right\rangle \\ &  \quad +4 \alpha^{-1} \left \langle \phi_0 \left| P\cdot P_f \right| P\cdot t(P_f/\alpha) \phi_0 \right\rangle + o(P^2) 
\end{align}
for small $|P|$. 
In particular, in combination with the norm \eqref{norm} above, we obtain for the inverse effective mass
\begin{align}\nonumber
\frac 1 {2m} & \leq \lim_{P\to 0} \frac 1{P^2} \left( \frac {  \left\langle \phi_P \left| H_P \right| \phi_P\right\rangle} {\| \phi_P\|^2} - E_0 \right) 
\\  & \leq 1 +  \frac 1{3\alpha^2}  \left\langle t(P_f/\alpha) \phi_0 \left| H_0 - E_0 \right|  t(P_f/\alpha)  \phi_0 \right\rangle  + \frac 4 {3\alpha}   \left \langle \phi_0 \left| P_f \cdot  t(P_f/\alpha) \right| \phi_0 \right\rangle \label{iem}
\end{align}
where we used again the rotation invariance of $\phi_0$. 

Our goal is to find a function $t$ such that 
 the right side of the above inequality goes to zero as $\alpha\to \infty$. To be precise, we shall find, for any $\delta>0$, a function $t$ such that the limit of the right side of \eqref{iem} is smaller than $\delta$, which is sufficient for our purpose. The following lemma, characterizing properties of the ground state of $H_0$ in the strong coupling limit,  will turn out to be essential. 

Let $\psi^{\rm P}$ be a minimizer of the Pekar functional in \eqref{def:EP}. As shown in \cite{Li}, it is unique up to translations and multiplication by a complex phase factor. We choose the phase factor such that $\psi^{\rm P}$ is non-negative, and translate the function to be rotation-invariant about the origin. Under these conditions, $\psi^{\rm P}$ is indeed unique. Let $\varphi^{\rm P}$ be the associated polarization field, given by \eqref{def:varphi} for $\alpha=1$. Note that both $\widehat\psi^{\rm P}$ and  $\varphi^{\rm P}$ are real-valued since $\psi^{\rm P}$ is an even function. Then the following holds.

\begin{lemma}\label{lem1}
Let $g:\R^3 \to \R$ be a smooth function with bounded second derivative.  With $\phi_0$ the ground state of $H_0$, we have 
\begin{equation}\label{lem:eq1}
\lim_{\alpha\to \infty}  \langle \phi_0 | g(P_f/\alpha) | \phi_0 \rangle = \int_{\R^3} | \widehat\psi^{\rm P} |^2 g\,.
\end{equation}
Moreover, if in addition $g$ is bounded,
\begin{equation}
\lim_{\alpha\to \infty} \alpha^{-2} \langle \phi_0 |  \N \, g(P_f/\alpha) | \phi_0 \rangle = \int_{\R^3} (\varphi^{\rm P})^2 \, \int_{\R^3} | \widehat\psi^{\rm P} |^2 g
\end{equation}
and, for any  $\xi \in L^2(\R^3)$, 
\begin{align}\nonumber
& \lim_{\alpha\to \infty} \alpha^{-1} \langle \phi_0 |   g(P_f/\alpha) a^\dagger (\xi_\alpha) g(P_f/\alpha) | \phi_0 \rangle \\ & = \iint\nolimits_{\R^3\times \R^3} dk\, dp\, \varphi^{\rm P}(k) \xi(k) \widehat\psi^{\rm P}(p+k) g(p+k) \widehat\psi^{\rm P} (p) g(p)  \label{lem:eq3}
\end{align}
where  $\xi_\alpha(p) = \alpha^{-3/2} \xi(p/\alpha)$ and we used the notation \eqref{adnot} for $a^\dagger(\xi_\alpha)$. 
\end{lemma}

In particular,  Lemma~\ref{lem1} states that the relevant expectation values can, in the strong coupling limit, be computed using the ansatz \eqref{ansatz} for $\phi_0$.

We shall postpone the proof of Lemma~\ref{lem1} to the end of this section, and continue by exploring its consequences. 
From \eqref{lem:eq1}, we obtain 
\begin{equation}\label{cons1}
\lim_{\alpha\to \infty} \alpha^{-1} \left \langle \phi_0 \left| P_f \cdot  t(P_f/\alpha)  \right| \phi_0 \right\rangle = \int_{\R^3} dp\,  \widehat\psi^{\rm P}(p)^2  p\cdot  t(p) \,.  
\end{equation}
We shall choose\footnote{When comparing with \eqref{compe}, note that $\nabla\widehat \psi^{\rm P}(p) = - \nabla\widehat\psi^{\rm P}(-p)$ since $\widehat\psi^{\rm P}$ is even.} 
\begin{equation}\label{def:t}
t(p) = \frac{   \nabla \widehat\psi^{\rm P}(p) }{\widehat\psi^{\rm P}(p)} \chi( \epsilon p))
\end{equation}
for some $\epsilon>0$, with $\chi$ a radial function in $C_0^\infty(\R^3)$ satisfying $\chi(0)=1$. In \cite{Li} it was shown that $\psi^{\rm P}$ is a smooth function that decays exponentially at infinity. In particular, this implies that   $\widehat\psi^{\rm P}$ and all its derivatives are bounded functions going to zero at infinity.  Moreover, from the variational principle \eqref{def:EP} it is not difficult to see that $\widehat\psi^{\rm P}$ is strictly positive. Hence the function $t$ in \eqref{def:t} is bounded and smooth for any $\epsilon>0$.  In particular, the assumptions in Lemma~\ref{lem1} are satisfied, and by combining \eqref{cons1} and \eqref{def:t}, we have
\begin{equation}
\lim_{\epsilon\to 0} \lim_{\alpha\to \infty} \alpha^{-1} \left \langle \phi_0 \left| P_f \cdot  t(P_f/\alpha)  \right| \phi_0 \right\rangle = \int_{\R^3} dp\,  \widehat\psi^{\rm P}(p)   p\cdot  \nabla \widehat\psi^{\rm P}(p) = -\frac 32
\end{equation}
where we used dominated convergence for the $\epsilon\to 0$ limit, and integrated by parts in the last step.

Theorem~\ref{main:thm} is thus proved if we can show that
\begin{equation}\label{tp}
\lim_{\epsilon\to 0} \lim_{\alpha\to \infty}  \alpha^{-2} \left\langle t(P_f/\alpha) \phi_0 \left| H_0 - E_0 \right|  t(P_f/\alpha)  \phi_0 \right\rangle = 3\,.  
\end{equation}
For the terms $P_f^2$, $\N$ and $E_0$ we can use again Lemma~\ref{lem1}, with the result that
\begin{align}\nonumber
& \lim_{\epsilon\to 0} \lim_{\alpha\to \infty}  \alpha^{-2} \left\langle t(P_f/\alpha) \phi_0 \left| P_f^2 + \N - E_0 \right|  t(P_f/\alpha)  \phi_0 \right\rangle \\ & = \int_{\R^3} dp\, |\nabla \widehat\psi^{\rm P}(p)|^2 \left( p^2 + \int_{\R^3} (\varphi^{\rm p})^2 - e^{\rm P}\right)   \label{tp1}
\end{align}
where we also used \eqref{plim}. 
In order to calculate  the expectation of  $V$, we cannot directly apply \eqref{lem:eq3} since the function $k\mapsto |k|^{-1}$ is not in $L^2(\R^3)$.  We shall introduce an ultraviolet cutoff $\Lambda>0$ and write
\begin{equation}\label{def:Lambda}
\frac 1 {|k|} =   \alpha^{-1} v( k /\alpha)  +  \frac{\theta(|k|-\Lambda \alpha)}{|k|}
\end{equation}
where $\theta$ denotes the Heaviside step function. Thus $v (k) = |k|^{-1} \theta (\Lambda - |k|)$, which is a function in $L^2(\R^3)$. After inserting the second term in \eqref{def:Lambda} into \eqref{def:V}, we can proceed as in the derivation of \cite[Eq.~(4) in Erratum]{LT} to obtain
\begin{equation}\label{214}
\pm \frac{ \sqrt{\alpha}}{\sqrt{2} \pi} \int_{|k|\geq \Lambda\alpha } dk\, \frac 1{|k|} \left(  a_k +  a_k^\dagger \right) \leq  \frac{8 \kappa}{3\pi \Lambda} P_f^2 + \frac 1 \kappa \left( \N + \frac 32\right)
\end{equation}
for any $\kappa>0$. 
Applying Lemma~\ref{lem1} and sending $\Lambda\to\infty$ followed by $\kappa\to \infty$, we conclude that 
\begin{equation}\label{tp2}
\lim_{\epsilon\to 0}\lim_{\alpha\to \infty} \left\langle t(P_f) \phi_0 \left| V \right|  t(P_f)  \phi_0 \right\rangle 
 =  - \frac{\sqrt 2}\pi \iint_{\R^3\times \R^3} dk\, dp\, \frac {\varphi^{\rm P}(k) } {|k|} \nabla\widehat\psi^{\rm P}(p+k)  \nabla\widehat\psi^{\rm P} (p) \,.
\end{equation}

The Euler-Lagrange equation satisfied by the minimizer of the Pekar functional \eqref{def:EP} reads in momentum space
\begin{equation}
(p^2 + \mu) \widehat\psi^{\rm P}(p) - \frac{ \sqrt{2}}\pi \int_{\R^3} dk\, \frac {\varphi^{\rm P}(k)}{|k|}  \widehat\psi^{\rm P}(p+k)  = 0
\end{equation}
with $\mu =  \int_{\R^3} (\varphi^{\rm P})^2 - e^{\rm P} $. Taking a derivative with respect to $p$, this becomes
\begin{equation}
(p^2 + \mu) \nabla \widehat\psi^{\rm P}(p) - \frac{\sqrt{2}}\pi \int_{\R^3} dk\, \frac {\varphi^{\rm P}(k)}{|k|}  \nabla \widehat\psi^{\rm P}(p+k)  = -2 p \widehat\psi^{\rm P}(p) \,.
\end{equation}
In particular, multiplying this equation by $\nabla \psi^{\rm P}(p)$ and integrating, we conclude that 
\begin{align}\nonumber
 & \int_{\R^3} dp\,  |\nabla \widehat\psi^{\rm P}(p)|^2  \left( p^2  + \mu \right) - \frac{\sqrt{2}} \pi \iint\nolimits_{\R^3\times \R^3} dp\, dk\,  \frac {\varphi^{\rm P}(k) } {|k|} \nabla\widehat\psi^{\rm P}(p+k)  \nabla\widehat\psi^{\rm P} (p) \\ & = - 2 \int_{\R^3} dp\,  \widehat\psi^{\rm P}(p)  p \cdot \nabla   \widehat\psi^{\rm P}(p) = 3 \,.
\end{align}
In combination with \eqref{tp1} and \eqref{tp2},  the identity \eqref{tp} follows, and consequently also the statement of  Theorem~\ref{main:thm}.

We are left with the 

\begin{proof}[Proof of Lemma~\ref{lem1}]
The key idea in the proof of Lemma~\ref{lem1} is to reintroduce the electron coordinate, and to  redo the proof of the strong coupling limit in  \cite{LT} with suitable perturbation terms. In fact, for  $\vec\lambda =(\lambda_1,\lambda_2,\lambda_3) \in \R^3$, we shall derive a lower bound on 
\begin{equation}
E_0(\vec\lambda) = \infspec  H_0 (\vec \lambda) 
\end{equation}
where $H_0(\vec\lambda)$ denotes the perturbed Hamiltonian 
\begin{align}\nonumber
H_0(\vec\lambda) &= H_0 + \lambda_1 \alpha^2 g_1(P_f/\alpha) + \lambda_2 g_2(P_f/\alpha)\N \\ & \quad  + \lambda_3 \alpha g_3(P_f/\alpha) \left( a(\xi_\alpha) + a^\dagger(\xi_\alpha)\right) g_3(P_f/\alpha)
\end{align}
for smooth, real-valued functions $g_i$, $1\leq i\leq 3$. We assume that the $g_i$ have bounded second derivative, and in addition that $g_2$ and $g_3$ are bounded. Under these assumptions, the perturbation terms are relatively form-bounded with respect to $H_0$, and hence $E_0(\vec\lambda)$ is finite for $|\vec\lambda|$ small enough. Moreover, since $E_0$ is a simple eigenvalue of $H_0$ that is isolated from the rest of the spectrum \cite{Mo}, $E_0(\vec\lambda)$ is differentiable for small $|\vec\lambda|$. 

We shall prove that as long as $|\lambda_2| \|g_2\|_\infty <1$, 
\begin{equation}\label{lem:claim}
\liminf_{\alpha\to \infty} \alpha^{-2} E_0(\vec\lambda) \geq E^{\rm P}(\vec\lambda) 
\end{equation}
where $E^{\rm P}(\vec\lambda)$ is the infimum of the perturbed Pekar functional 
\begin{align}\nonumber
  \mathcal{E}^{\rm P}(\psi,\varphi) & + \lambda_1 \int_{\R^3} g_1 |\widehat\psi|^2 + \lambda_2 \int_{\R^3} |\varphi|^2 \int_{\R^3} g_2 |\widehat\psi|^2  \\  &    + 2 \lambda_3 \Re  \iint_{\R^3\times\R^3} dk\, dp\,  \varphi (k) \xi(k) \widehat\psi^*(p+k) g_3(p+k) \widehat\psi (p) g_3(p)   
 \end{align}
 (subject to the normalization condition $ \int_{\R^3} |\psi|^2 =1$), with 
 $\mathcal{E}^{\rm P}(\psi,\varphi) $ denoting the Pekar functional
\begin{equation}
\mathcal{E}^{\rm P}(\psi,\varphi)  = \int_{\R^3} dx\, |\nabla\psi(x)|^2  - \frac{ \sqrt{2\alpha}} \pi \Re  \int_{\R^3} dk \, \frac{\varphi(k)}{|k|}  \int_{\R^3} dx\, |\psi(x)|^2 e^{ik\cdot x} + \int_{\R^3} dp\,  |\varphi(p)|^2  \,.
\end{equation}
We note that also $E^{\rm P}(\vec \lambda)$ is finite for $|\vec\lambda|$ small enough. Moreover, the uniqueness of minimizers of $\mathcal {E}^{\rm P}$ (up to translations and multiplication by a complex phase) implies that $E^{\rm P}(\vec \lambda)$ is differentiable at $\vec\lambda = \vec 0$.

The derivative of $E_0(\vec \lambda)$ at $\vec\lambda=\vec 0$ equals
\begin{align}\nonumber
\vec\lambda \cdot \nabla E_0(0) & =  \lambda_1 \alpha^2  \langle \phi_0 | g_1(P_f^2 /\alpha) |\phi_0\rangle  + \lambda_2  \langle \phi_0 | \N g_2(P_f^2 /\alpha) |\phi_0\rangle \\ & \quad + 2 \lambda_3 \alpha\Re
 \langle \phi_0 |   g_3(P_f/\alpha) a^\dagger(\xi_\alpha) g_3(P_f/\alpha) | \phi_0 \rangle \,.
 \end{align}
Moreover, from the concavity of $E_0(\vec\lambda)$ we have 
\begin{equation}
 \vec\lambda \cdot \nabla E_0(0) \geq E_0(\vec\lambda) - E_0
\end{equation}
and hence \eqref{lem:claim} implies that 
\begin{align}\nonumber
& \liminf_{\alpha\to \infty} \left[  \lambda_1 \ \langle \phi_0 | g_1(P_f^2 /\alpha) |\phi_0\rangle  + \lambda_2 \alpha^{-2}  \langle \phi_0 | \N g_2(P_f^2 /\alpha) |\phi_0\rangle \right. \\ \nonumber & \qquad\quad  \left. + 2 \lambda_3  \alpha^{-1} \Re
 \langle \phi_0 |   g_3(P_f/\alpha) a^\dagger(\xi_\alpha) g_3(P_f/\alpha) | \phi_0 \rangle \right]
\\ &   \geq E^{\rm P}(\vec\lambda) - E^{\rm P}(\vec 0)\,, \label{bsi}
\end{align}
where we have used \eqref{plim} and the fact that $E^{\rm P}(\vec 0) = e^{\rm P}$. 
Both sides of \eqref{bsi} are concave function of $\vec\lambda$ that vanish at $\vec\lambda=\vec 0$. Since the right side is differentiable at $\vec\lambda=\vec 0$, the same  holds for the left side, and the two derivatives agree. We conclude that the limits $\alpha\to\infty$ of the various terms actually exist, and satisfy
\begin{align}\nonumber
& \lambda_1 \lim_{\alpha\to\infty} \langle \phi_0 | g_1(P_f^2 /\alpha) |\phi_0\rangle  + \lambda_2  \lim_{\alpha\to\infty} \alpha^{-2}  \langle \phi_0 | \N g_2(P_f^2 /\alpha) |\phi_0\rangle \\ \nonumber & + 2 \lambda_3  \lim_{\alpha\to\infty} \alpha^{-1} \Re
 \langle \phi_0 |   g_3(P_f/\alpha) a^\dagger(\xi_\alpha) g_3(P_f/\alpha) | \phi_0 \rangle 
\\ &  = \vec\lambda \cdot \nabla E^{\rm P}(\vec 0) \,.
\end{align} 
In particular, 
\begin{equation}
  \lim_{\alpha\to \infty}   \langle \phi_0 | g_1(P_f /\alpha) |\phi_0\rangle   = \nabla_{\lambda_1}{E^{\rm P}}(\vec 0) =  \int_{\R^3} g_1  |\widehat\psi^{\rm P}|^2
\end{equation}
\begin{equation}
  \lim_{\alpha\to \infty} \alpha^{-2}  \langle \phi_0 | \N\, g_2(P_f /\alpha) |\phi_0\rangle   = \nabla_{\lambda_2}{E^{\rm P}}(\vec 0) =  \int_{\R^3} (\varphi^{\rm P})^2 \int_{\R^3} g_2  |\widehat\psi^{\rm P}|^2
\end{equation}
and
\begin{align}\nonumber
 & \lim_{\alpha\to\infty} \alpha^{-1}  \Re
  \langle \phi_0 |   g_3(P_f/\alpha) a^\dagger(\xi_\alpha) g_3(P_f/\alpha) | \phi_0 \rangle  \\ &=  \frac 12 \nabla_{\lambda_3}{E^{\rm P}}(\vec 0) = \Re \iint_{\R^3\times\R^3} dk\,dp\,  \varphi^{\rm P}(k) \xi(k) \psi^{\rm P}(p+k) g_3(p+k) \psi^{\rm P} (p) g_3(p)  \,.
\end{align}
By linearity in $\xi$, the corresponding identity for the imaginary part follows by replacing $\xi$ by $i \xi$. Hence the desired statements \eqref{lem:eq1}--\eqref{lem:eq3} are proved.

It remains to derive the claimed lower bound \eqref{lem:claim} on $E_0(\vec\lambda)$. We note that $H_0 (\vec \lambda)$ is the restriction to total momentum equal to zero of the translation-invariant operator
\begin{align}\nonumber
\mathfrak{h}_{\vec\lambda} &= - \Delta  - \frac{ \sqrt{\alpha}}{\sqrt{2}\pi} \int_{\R^3} dk\, \frac 1{|k|} \left( e^{i k x} a_k + e^{-ikx} a_k^\dagger \right) + \N
\\ \nonumber & \quad + \lambda_1\alpha^2 g_1 ( -i\alpha^{-1} \nabla) +  \lambda_2\N \, g_2 ( -i\alpha^{-1} \nabla) \\ & \quad + \lambda_3 \alpha g_3(-i \alpha^{-1} \nabla) \int_{\R^3} dk \left( e^{ikx} \xi^*_\alpha(k) a_k + e^{-ikx} \xi_\alpha(k) a^\dagger_k \right) g_3(-i\alpha^{-1} \nabla)
\end{align}
acting on $L^2(\R^3) \otimes \mathcal{F}$.  In particular, $E_0(\vec\lambda) \geq \infspec \mathfrak{h}_{\vec\lambda}$. 

To derive a lower bound on $\infspec \mathfrak{h}_{\vec\lambda}$, we proceed as in \cite{LT}. The first step is to introduce an ultraviolet cutoff in the interaction $V$.  Similarly to \eqref{214} above, we have
\begin{equation}
 \frac{ \sqrt{\alpha} }{\sqrt{2}\pi} \int_{|k|\geq \Lambda\alpha} dk\, \frac 1{|k|} \left( e^{i k x} a_k + e^{-ikx} a_k^\dagger \right) \leq - \frac{8\kappa}{3\pi \Lambda} \Delta + \frac 1 \kappa \left(\int_{|k|\geq \Lambda \alpha} dk\, a^\dagger_k a_k  + \frac 32\right)
\end{equation}
for $\kappa > 0$. This was proved in   \cite[Eq.~(4) in Erratum]{LT} (where $\kappa=1$ was chosen). 
Hence we can introduce an ultraviolet cutoff $\Lambda\alpha$ on the phonon modes, with small error  as long as $\Lambda \gg 1$. For the last term in $\mathfrak{h}_{\vec\lambda}$ multiplying $\lambda_3$, we simply use 
\begin{align}\nonumber 
& \pm \int_{|k|\geq \Lambda\alpha} dk \left( e^{ikx} \xi^*_\alpha(k) a_k + e^{-ikx} \xi_\alpha(k) a^\dagger_k \right) \\ \nonumber & \leq 2 \left( \int_{|k|\geq \Lambda} dk\, |\xi(k)|^2\right)^{1/2} \left(\int_{|k|\geq \Lambda \alpha} dk\, a^\dagger_k a_k  + \frac 12\right) ^{1/2}
\\& \leq \frac \alpha \epsilon  \int_{|k|\geq \Lambda} dk\, |\xi(k)|^2 + \frac \epsilon \alpha \left( \int_{|k|\geq \Lambda \alpha} dk\, a^\dagger_k a_k  + \frac 12\right) 
\end{align}
for any $\epsilon>0$. Again this term only introduces a small error if $\Lambda$ is large.

In particular, if we choose $\kappa$ and $\epsilon$ such that 
\begin{equation}
\frac 1 \kappa + |\lambda_2|  \|g_2\|_\infty + \epsilon |\lambda_3| \|g_3\|_\infty^2 \leq 1   
\end{equation}
we have
\begin{equation}
\infspec \mathfrak{h}_{\vec\lambda} \geq \infspec \mathfrak{h}_{\vec\lambda}^{(1)} - \frac 3{2\kappa} - |\lambda_3| \alpha^2 \|g_3\|_\infty^2 \epsilon^{-1}   \int_{|k|\geq \Lambda} dk\, |\xi(k)|^2
\end{equation}
where
\begin{align}\nonumber
\mathfrak{h}_{\vec\lambda}^{(1)} & = - \left( 1 - \frac{8\kappa}{3\pi \Lambda}\right)  \Delta   - \frac{ \sqrt{\alpha}}{\sqrt{2}\pi}  \int_{|k|\leq \Lambda\alpha} dk\, \frac 1{|k|} \left( e^{i k x} a_k + e^{-ikx} a_k^\dagger \right) + \N  \\ \nonumber &  \quad + \lambda_1\alpha^2 g_1 ( -i\alpha^{-1} \nabla)  +  \lambda_2\N\, g_2 ( -i\alpha^{-1} \nabla)  \\ & \quad + \lambda_3 \alpha g_3(-i \alpha^{-1} \nabla) \int_{|k|\leq \Lambda\alpha} dk \left( e^{ikx} \xi^*_\alpha(k) a_k + e^{-ikx} \xi_\alpha(k) a^\dagger_k \right) g_3(-i\alpha^{-1} \nabla) \,.
\end{align}
Here $\N$ stands now for the number of phonons with momenta $|k|\leq \Lambda\alpha$. Equivalently, $\N$ could be taken to be the total particle number, without effecting the ground state energy of $\mathfrak{h}_{\vec\lambda}^{(1)}$, since $|\lambda_2|\|g_2\|_\infty<1$ by assumption, and hence occupying phonon modes with $|k|>\Lambda\alpha$ raises the energy.

Next we shall localize the electron. With $\phi\in H^1(\R^3)$ a real-valued function of compact support, normalized such that $\int_{\R^3} \phi^2 =1$, let $\phi_y(x) = \phi(x-y)$. For any  $\Psi \in L^2(\R^3) \otimes \mathcal F$ of finite energy, we compute
\begin{align}\nonumber
&\int_{\R^3} dy \left\langle \phi_y \Psi \left| \mathfrak{h}_{\vec\lambda}^{(1)} \right| \phi_y \Psi \right\rangle  \\ \nonumber &= \left\langle  \Psi \left| \mathfrak{h}_{\vec\lambda}^{(1)} \right|  \Psi \right\rangle +\left( 1 - \frac{8\kappa}{3\pi \Lambda}\right)  \| \Psi \|^2\int_{\R^3} |\nabla\phi|^2 \\ \nonumber &\quad + \lambda_1 \alpha^2 \iint_{\R^3\times \R^3} dp\, dq\left(  g_1(p/\alpha) - g_1(q/\alpha) \right) |\widehat \phi(p-q)|^2 \| \widehat\Psi(q)\|_{\mathcal F}^2\\ \nonumber &\quad + \lambda_2  \iint_{\R^3\times \R^3} dp\, dq\left(  g_2(p/\alpha) - g_2(q/\alpha) \right) |\widehat \phi(p-q)|^2 \| \N^{1/2}\widehat\Psi(q)\|_{\mathcal F}^2\\ \nonumber &\quad + 2 \lambda_3 \alpha \Re 
\iint_{\R^3\times \R^3} dp\, dq\, |\widehat \phi(p-q)|^2 \int_{|k|\leq \Lambda\alpha} dk \,  \xi_\alpha(k)  \left\langle \widehat\Psi(q) \right| a^\dagger_k \left|  \hat\Psi(q) \right\rangle_{\mathcal{F}}   \\ & \qquad\qquad\qquad\quad  \times  \left( g_3(p/\alpha)g_3((p+k)/\alpha) - g_3(q/\alpha) g_3((q+k)/\alpha)\right)  \,. \label{tbp}
\end{align}
Here $\widehat \Psi(q) \in \F$ denotes the Fock space vector obtained by fixing the electron momentum of $\Psi$ to be $q$. 
By assumption the functions $g_i$  have  bounded second derivatives. Therefore, 
\begin{equation}
\left| g_i(p/\alpha) - g_i(q/\alpha) - \alpha^{-1} \nabla g_i(q/\alpha) \cdot (p-q) \right| \leq  C_i \alpha^{-2} |p-q|^2
\end{equation}
for suitable constants $C_i >0$. Moreover, since $g_3$ is in addition assumed to be bounded, we also have
\begin{align}\nonumber
& \left|  g_3(p/\alpha)g_3((p+k)/\alpha) - g_3(q/\alpha) g_3((q+k)/\alpha) \right. \\  & \left.  -  \alpha^{-1} \left[ \nabla g_3(q/\alpha)g_3((q+k)/\alpha) +  g_3(q/\alpha)\nabla g_3((q+k)/\alpha)\right] \cdot (p-q) \right| \leq  C_3 \alpha^{-2} |p-q|^2
\end{align}
for some constant $C_3>0$ independent of $k$. We plug these bounds into \eqref{tbp}, and use that $\int_{\R^3} dp\, p |\widehat \phi(p)|^2 = 0$,  $\int_{\R^3} dq\, \|\widehat \Psi(q)\|_{\mathcal{F}}^2 = 1$ as well as
\begin{equation}
\int_{\R^3} dq  \int_{|k|\leq \Lambda\alpha} dk \,  |\xi_\alpha(k)|  \left| \left\langle \widehat\Psi(q) \right| a^\dagger_k \left|  \hat\Psi(q) \right\rangle_{\mathcal{F}} \right| \leq   \|\xi\|_2 \left\langle\Psi \left| \sqrt{\N + \tfrac 12} \right| \Psi \right\rangle\,.
\end{equation}
This way we obtain the bound
\begin{align}\nonumber
\int_{\R^3} dy \left\langle \phi_y \Psi \left| \mathfrak{h}_\lambda^{(1)} \right| \phi_y \Psi \right\rangle &  \leq  \left\langle  \Psi \left| \mathfrak{h}_\lambda^{(1)} \right|  \Psi \right\rangle +\left( 1 - \frac{8\kappa}{3\pi \Lambda} + C_1 |\lambda_1| \right)\| \Psi \|^2 \int_{\R^3} |\nabla\phi|^2 \\ \nonumber & \quad + C_2 |\lambda_2| \alpha^{-2} \| \N^{1/2} \Psi\|^2\int_{\R^3} |\nabla\phi|^2 \\ & \quad +  2C_3 |\lambda_3| \alpha^{-1}    \|\xi\|_2 \left\langle\Psi \left| \sqrt{\N + \tfrac 12} \right| \Psi \right\rangle \int_{\R^3} |\nabla\phi|^2 \,.
\end{align}
Since 
\begin{equation}
\int_{\R^3} dy \left\|  \phi_y \Psi  \right\|^2 = \| \Psi \|^2 
\end{equation}
holds for any $\Psi\in L^2(\R^3)\otimes\mathcal{F}$, 
we can find, for any given $\Psi$, a $y\in \R^3$ such that 
\begin{equation}
 \left\|  \phi_y \Psi  \right\|^{-2}   \left\langle \phi_y \Psi \left| \mathfrak{h}_\lambda^{(2)} \right| \phi_y \Psi \right\rangle  \leq   \left\|  \Psi  \right\|^{-2}  \left\langle  \Psi \left| \mathfrak{h}_\lambda^{(1)} \right|  \Psi \right\rangle 
\end{equation}
where
\begin{align}\nonumber
\mathfrak{h}_{\vec\lambda}^{(2)} &= \mathfrak{h}_{\vec\lambda}^{(1)} -  \left( 1 - \frac{8\kappa}{3\pi \Lambda} + C_1 |\lambda_1| \right) \int_{\R^3} |\nabla\phi|^2  - C_2 |\lambda_2| \alpha^{-2} \N  \int_{\R^3} |\nabla\phi|^2 \\ & \quad -  2 C_3 |\lambda_3| \alpha^{-1}  \|\xi\|_2 \sqrt{\N + \tfrac 12}  \int_{\R^3} |\nabla\phi|^2\,. \label{h2}
\end{align}
In particular, to obtain a lower bound on the ground state energy of $\mathfrak{h}_{\vec\lambda}^{(1)} $,  we can minimize the expectation value of $\mathfrak{h}_{\vec\lambda}^{(2)} $ over functions $\Psi$ with electron coordinate supported in a ball of radius $R$. By translation invariance, we may assume without loss of generality that this ball is centered at the origin.
The relative error in the energy coming from the additional terms in \eqref{h2} is of the order $\int_{\R^3} |\nabla\phi|^2 \sim R^{-2}$, which is much less than $ \alpha^2$ if we choose $R\gg \alpha^{-1}$. 

The remainder of the proof is now identical to \cite{LT}, and we will skip the details. With both an ultraviolet cutoff (for the phonon momenta) and a space cutoff (for the electron) in place, one can  approximate the interaction terms with finitely many modes, and use coherent states to compare the Hamiltonian to the corresponding classical problem, yielding the Pekar energy. 
This yields \eqref{lem:claim}, and hence completes the proof of Lemma~\ref{lem1}. 
\end{proof}

\bigskip
\noindent {\bf Acknowledgments.} 
Financial support through the European Research Council (ERC) under the European Union's Horizon 2020 research and innovation programme (grant agreement No 694227; R.S.)  is gratefully acknowledged.


\begin{thebibliography}{19}

\bibitem{devreese} A.S. Alexandrov, J.T. Devreese, {\it Advances in Polaron Physics}, Springer (2010).

\bibitem{DoVa} M. Donsker, S.R.S. Varadhan, \textit{Asymptotics for the polaron}, Comm. Pure Appl. Math. \textbf{36}, 505--528 (1983).

\bibitem{FLST} R.L. Frank, E.H. Lieb, R. Seiringer, L.E. Thomas, {\it Ground state properties of multi-polaron systems},  in: XVIIth International Congress on Mathematical Physics, Proceedings of the ICMP held in Aalborg, August 6--11, 2012, A. Jensen (ed.), 477--485, World Scientific, Singapore (2013).

\bibitem{Fr} H. Fr\"ohlich, \textit{Theory of electrical breakdown in ionic crystals}, Proc. R. Soc. Lond. A \textbf{160},  230--241 (1937).

\bibitem{GL} B. Gerlach, H. L\"owen, {\it Analytical properties of polaron systems or: Do
polaronic phase transitions exist or not?}, Rev. Mod. Phys. {\bf 63}, 63--90 (1991).

\bibitem{GW}  M. Griesemer, A. W\"unsch, {\it 
Self-Adjointness and Domain of the Fr\"ohlich Hamiltonian},  J. Math. Phys. {\bf 57}, 021902 (2016).

\bibitem{pekar} L.D. Landau, S.I. Pekar, {\it Effective Mass of a Polaron}, Zh. Eksp. Teor. Fiz. {\bf 18}, 419--423 (1948).

\bibitem{LLP} T.D. Lee, F.E. Low, D. Pines, {\it The motion of slow electrons in a polar crystal},  Phys. Rev. {\bf 90}, 297--302 (1953).

\bibitem{Li} E. H. Lieb, \textit{Existence and uniqueness of the minimizing solution of Choquard's nonlinear equation}, 
Studies in Appl. Math. \textbf{57}, no. 2, 93--105 (1976/77).

\bibitem{LS14} E.H. Lieb, R. Seiringer, {\it Equivalence of Two Definitions of the Effective Mass of a Polaron}, J. Stat. Phys. {\bf 154}, 51--57 (2014). 

\bibitem{LT} E.H. Lieb, L.E. Thomas, {\it Exact Ground State Energy of the Strong-Coupling Polaron}, Commun. Math. Phys. {\bf 183}, 511--519 (1997); {\it Erratum}, Commun. Math. Phys. {\bf 188}, 499--500 (1997).

\bibitem{Mo} J.S. M\o ller, \textit{The polaron revisited}, Rev. Math. Phys. \textbf{18}, 485--517 (2006).

\bibitem{MV1}  C. Mukherjee, S.R.S. Varadhan, 
{\it Strong coupling limit of the Polaron measure and the Pekar process}, 
preprint arXiv:1806.06865 

\bibitem{MV2} C. Mukherjee, S.R.S. Varadhan, 
{\it Identification of the Polaron measure in strong coupling and the Pekar variational formula}, preprint arXiv:1812.06927 

\bibitem{nagy} P. Nagy, {\it A note to the translationally invariant strong coupling theory of the polaron}, Czech. J. Phys. B {\bf 39}, 353--356 (1989).

\bibitem{spohn} H. Spohn, {\it  Effective Mass of the Polaron: A Functional Integral Approach}, Ann. Phys. {\bf 175}, 278--318 (1987)

\end{thebibliography}
\end{document}